\documentclass[12pt,draftclsnofoot, onecolumn,journal]{IEEEtran}
\usepackage{times}
\usepackage[final]{graphicx}
\usepackage[reqno]{amsmath}
\usepackage{amsfonts}

\usepackage{times,amsmath,epsfig}
\usepackage{latexsym,amssymb}
\usepackage{cite}
\usepackage{subfigure}
\usepackage{stfloats}
\usepackage{amsmath}
\usepackage{amssymb}
\usepackage{array}
\usepackage[noend]{algorithmic}
\usepackage{algorithm}
\usepackage{color}
\usepackage{psfrag}
\usepackage{epsfig}
\definecolor{darkgreen}{rgb}{0,0.6,0}

\def\SNR{{\rm SNR}}

\begin{document}

\title{\huge{Power Allocation and Pricing in Multi-User Relay Networks Using Stackelberg and Bargaining Games} }

\author{Qian~Cao, H. Vicky~Zhao, and Yindi~Jing\\
ECE Dept., University of Alberta, Edmonton, AB Canada T6G 2V4\\
Email: qcao1@ualberta.ca, vzhao@ece.ualberta.ca, yindi@ualberta.ca}

\maketitle

\begin{abstract}
This paper considers a multi-user single-relay wireless network,
where the relay gets paid for helping the users forward signals, and the users pay to receive the relay service.
We study the relay power allocation and pricing problems, and model
the interaction between the users and the relay as a two-level
Stackelberg game.
In this game, the relay, modeled as the service provider and the
leader of the game, sets the relay price to maximize its revenue;
while the users are modeled as customers and the followers who buy
power from the relay for higher transmission rates. We use a
bargaining game to model the negotiation among users to
achieve a fair allocation of the relay power. Based on the proposed
fair relay power allocation rule, the optimal relay power price that
maximizes the relay revenue is derived analytically. Simulation
shows that the proposed power allocation scheme achieves higher
network sum-rate and relay revenue than the even power allocation.
Furthermore, compared with the sum-rate-optimal solution, simulation
shows that the proposed scheme achieves better fairness with
comparable network sum-rate for a wide range of network scenarios.
The proposed pricing and power allocation solutions are also shown
to be consistent with the laws of supply and demand.

\end{abstract}

\section{Introduction}
\label{sec-intro} Cooperative communication has been shown to be a
promising concept for future wireless networks (e.g.,
\cite{*oview,*promis}). The basic idea is to have multiple
nodes in the network help each other's transmission to achieve diversity. 
Numerous cooperative relaying strategies aiming at optimizing the
global network performance have been proposed. Two widely used ones
are amplify-and-forward (AF) and decode-and-forward (DF). For AF,
the relays simply forward amplified versions of the signals they
receive without decoding. For DF, the relays decode their received
signals and transmit the re-encoded signals \cite{*Laneman}. While
pioneering efforts in cooperative networks focus on single-user
networks (e.g., \cite{*111,*222,*333,*3333,*444,*555}), research on
multi-user networks, in which multiple transmissions of different
users can be supported, are important to meet the demands of future
communication systems. One such model is the multi-user single-relay
network (e.g.,
\cite{*MARN1,*MARN3,*LLB1,*LLB2,*Tairan,*Sha,*Guftaar}). Capacity
bounds of the network are investigated in \cite{*MARN1,*MARN3}, and
interference cancellation schemes are proposed in \cite{*LLB1,
*LLB2}. In \cite{*Tairan} and \cite{*Sha}, network decoding is
applied to combat interference among users in multi-user
single-relay networks. In \cite{*Guftaar}, the resource allocation
problem, including both the subcarrier allocation  and the relay
power allocation, in a multi-user single-relay network is studied to
maximize the sum-rate.

In the aforementioned papers, nodes in a network are assumed to be altruistic and
willing to cooperate to optimize the overall network performance. In
many practical applications, however, nodes are selfish and aim to
optimize their own benefits or quality-of-service.
On the other hand, due to limited resources available, multiple
selfish users in the network have conflicts in resource allocation.
It is important to find solutions that achieve both high overall
network performance and fairness among users. To model and analyze
these behaviors, game theory is often used \cite{*GTWireless}.

Research in cooperative network designs using game theory become
increasingly popular in recent years. Many papers have appeared in
literature, e.g., \cite{*slr,*slr2,*bei,add1,add2,*3,*4,*5}.
\cite{*4} considers the uplink of a network with multiple users and
single base station, in which users can form coalitions and share
resources to form a virtual multi-antenna system to improve
performance. A merge-and-split algorithm is proposed to construct
coalitions among users to maximize utilities, which are defined as
the transmission rates. In \cite{*3} and \cite{*5}, a two-user
network where each user can also work as a relay for the other is
studied. By employing a two-user bargaining game, fair bandwidth
allocation \cite{*3} and power allocation \cite{*5} are found from
the Nash bargaining solution. In \cite{*slr}, the relay power
allocation problem in the downlink of a multi-user multi-relay
cellular network is studied. Non-cooperative game theory is used to
model the competition for relay power among users. An iterative
scheme is proposed to ensure all users reach the Nash equilibrium
point.

None of the aforementioned works consider the cooperation
stimulation problem.
In a relay network, one possible solution to user cooperation
stimulation is the payment-based scheme, where the relays get paid
if they forward users' messages and users pay for the relay service.
Here, relays and users aim to maximize their own payoffs instead of
the overall network performance. In \cite{*bei,*slr2,add1,add2}, the pricing
mechanism is used to encourage the relays to help forward packets
for the users. For multi-user ad hoc networks, compensation
frameworks are proposed in \cite{add1} and \cite{add2}, in which a
network node sets price and receives revenue if it cooperatively
helps others' transmissions. In \cite{*bei}, for a single-user
multi-relay network, the relay selection and relay power control are
investigated using a two-level Stackelberg game. In this game, the
relays compete to provide service to the user to gain revenue.
\cite{*slr2} studies the user power control and relay pricing
problems in a multi-user single-relay network. In the game theoretic
model, the relay sets the price to maximize its revenue, while a
non-cooperative game is used to model the user behavior, in which
each user adjusts its transmit power to selfishly maximize its own
utility. A distributed iterative scheme is proposed to achieve the
unique Nash equilibrium point.

In this paper, we consider a multi-user single-relay network and use
game theory to model and analyze the user and relay behavior.
Amplify-and-forward is adopted due to its simplicity and low
computational load. Pricing mechanism is used where the relay gets
paid for signal forwarding and users pay for the relay service. We
model the interaction between the relay and the users as a two-level
Stackelberg game, in which the relay is the leader and sets the unit
power price for the relay service, and users are the followers where
each user decides how much power to purchase from the relay.
Different from \cite{*slr2}, we consider the relay power allocation
among users, instead of the user power control; and also we model
the relay power competition among users as a cooperative bargaining
game. For the relay power allocation, the Kalai-Smorodinsky
bargaining solution (KSBS) is used for fairness. The power allocation
problem is transformed into  a convex optimization
problem. With the KSBS-based relay power allocation, We
analytically find the optimal relay price that maximizes the relay
revenue. From our simulations, compared with the sum-rate-optimal
power allocation, the proposed KSBS-based power allocation is fairer
and achieves close-to-optimal sum-rate for a wide range of network
scenarios. Compared with the even power allocation, the proposed
KSBS-based power allocation archives higher relay revenue and
network sum-rate. It is also shown via simulations that the proposed
relay pricing and power allocation solutions are consistent with the
laws of supply and demand.


The rest of this paper is organized as follows. Section
\ref{sec-background} gives a brief background review on game theory.
Section \ref{sec-model} describes the multi-user single-relay
cooperative network model and the Stackelberg game and bargaining
game models for the relay pricing and relay power allocation
problems. In Section \ref{sec-solution}, we analyze the relay power
pricing and power allocation problems. The optimal relay price is
solved analytically, while the relay power allocation is transformed
into a convex optimization problem. Section \ref{sec-discussion}
discusses the properties of the proposed solutions and their
possible implementation. Simulation results are shown in Section
\ref{sec-simulation}. Conclusions are drawn in Section
\ref{sec-conclusion}.


\section{Review of Game Theory}
\label{sec-background} In this section, we provide a brief review of
game theory and introduce the basic concepts and results that will
be used in this paper.

Game theory mainly analyzes behavior in strategic situations, or
games, in which players influence each other's decision and
performance \cite{*Owen}. A game consists of three parts: a set of
players, a set of actions of the players, and a set of utilities
that represent the players' relative satisfaction of the outcome of
the game. Equilibrium is the strategy outcome of a game that is the
best response of each user given the decision of others. The most
famous equilibrium is the Nash equilibrium, in which no player can
increase its utility by unilaterally changing its own strategy, and
the corresponding strategy set and utilities constitute a
Nash equilibrium.


\subsection{Stackelberg Game}
\label{sec-backgroundsubB} In a Stackelberg game, one player acts as
a leader who takes action first,
and the other players are followers who observe the leader's action and act accordingly. 
The subgame perfect Nash equilibrium of the Stackelberg game can be
found using the backward induction method. It first studies the
followers' game: for each possible action of the leader, it finds
the optimal followers' response that maximizes the followers'
payoff. Then given the optimal followers' response strategy, it
studies the leader's action and chooses the one that maximizes the
leader's utility. The chosen strategy set is the Stackelberg
equilibrium \cite{*Owen}.

%


\subsection{Bargaining Game and Kailai-Smorodinsky Bargaining Solution (KSBS)}
\label{sec-backgroundsub-KSBS} A bargaining game models the
bargaining interactions of players. In a bargaining problem, there
are $N$ players with utilities $u_1,\ u_2,\ \cdots u_N$. A utility
vector $\mathbf{u}=(u_1\ u_2\ \cdots\ u_N)$ is called feasible if it
is possible to find a strategy set that gives the $i$th player
utility $u_i$ for all $i=1, \cdots, N$. Let $\mathcal {S}$ denote
the set of all feasible utility vectors. The disagreement point,
denoted as $\mathbf{u}_0=(u_{1,0}\ u_{2,0}\ \dots u_{N,0})$, is the
vector of the minimal utility that each player expects if the
players do not reach an agreement and play non-cooperatively. It is
the guaranteed utility for the players in the bargaining game. The
ideal point $\mathbf{u}^I=(u^I_1\ u^I_2\ \dots u^I_N)$ (the
superscript `$I$' stands for ideal) is the vector of the maximum
achievable utilities of the players in $\mathcal {S}$. We thus have
$u^I_i\geq u_{i,0}$. Note that for players with $u^I_i= u_{i,0}$,
cooperation does not increase their utilities and they will not
enter the game. For the rest of the players, $u^I_i> u_{i,0}$ and
they will participate in the bargaining process.

Given $\mathcal {S}$, the disagreement point $\mathbf{u}_0$, and the
ideal point $\mathbf{u}^I$,
players negotiate with each other to select one feasible solution
$\mathbf{u}$ and the corresponding strategy set that all players
agree. Depending on how players define fairness, they may select
different solutions. One popular solution for the bargaining game is
the KSBS \cite{*Bargaining}, which is the solution to the
optimization problem
\begin{eqnarray}
\max{k} \hspace{5mm} \mbox{s.t.}\quad
\frac{u_i-u_{i,0}}{u^I_i-u_{i,0}}=k \label{KSBS}
\end{eqnarray}
for all participating players with $u_i^I > u_{i,0}$.

KSBS is an equilibrium point that guarantees fairness in the sense
of equal penalty, which can be deducted from the constraint in
(\ref{KSBS}). In (\ref{KSBS}), $(u^I_i-u_{i,0})$ and $(u_i-u_{i,0})$
are Player $i$'s maximum and actual net utility gains, respectively.
Taking logarithm on both sides of the constraint in (\ref{KSBS}), we
have
\begin{equation}
    \log{(u^I_i-u_{i,0})}-\log{(u_i-u_{i,0})}=-\log k.
\end{equation}
As $\log k$ is a constant independent of the players, the constraint
in (\ref{KSBS}) forces all players participating in the bargaining
game to suffer the same quality penalty in the logarithmic scale,
and thus ensures fairness in this sense.
It is worth mentioning that KSBS is neither individual utility
optimal nor global optimal in general. It is an equilibrium point
that balances the proposed utility measure and fairness among users.



\section{Game Models for Relay Power Allocation and Pricing}\label{sec-model}

In this section, we explain the network model,  elaborate the relay
power pricing and relay power allocation problems, and propose the
game theoretical models for the problems using a Stackelberg game and
a bargaining game.

\subsection{Network Model}
Consider a wireless network with $N$ users communicating with their
destinations with the help of one relay as shown in Figure 1.
\begin{figure}[!ht]
  \centering
    \psfrag{hh1}{$h_1$}
    \psfrag{hh2}{$h_i$}
    \psfrag{hh3}{$h_N$}
    \psfrag{ff1}{$f_1$}
    \psfrag{ff2}{$f_i$}
    \psfrag{ff3}{$f_N$}
    \psfrag{gg1}{$g_1$}
    \psfrag{gg2}{$g_i$}
    \psfrag{gg3}{$g_N$}
    \psfrag{User 1}{{\small\hspace{-2mm} User 1}}
    \psfrag{User i}{{\small\hspace{-2.5mm} User $i$}}
    \psfrag{User N}{{\small\hspace{-2.3mm} User $N$}}
    \psfrag{Destination 1}{{\small Destination 1}}
    \psfrag{Destination i}{{\small Destination $i$}}
    \psfrag{Destination N}{{\small Destination $N$}}
  \includegraphics[width=.6\textwidth]{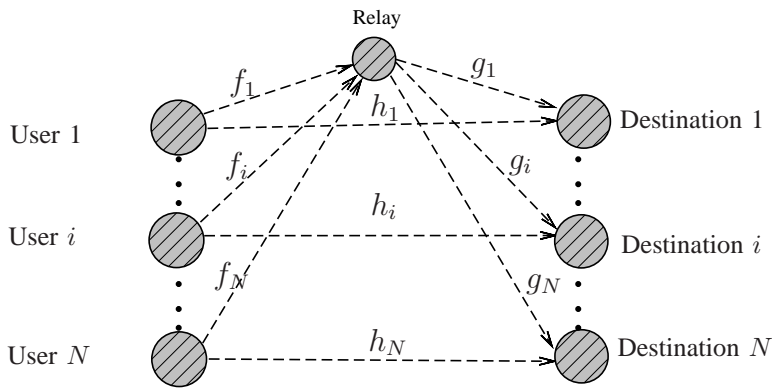}\\
  \label{fig-NWModel}
  \caption{Multi-user single-relay network.}\label{fig:2}
\end{figure}
Denote the channel gain from User $i$ to Destination $i$ (direct
link) as $h_i$, the channel gain from User $i$  to the relay as
$f_i$,  and the channel gain from the relay to Destination $i$ as
$g_i$. The relay and destination are assumed to have global and
perfect channel state information (CSI) $f_1,\cdots,f_N,g_1,\cdots,g_N$ through training and
feedback. No CSI is required at the users. User $i$ uses transmit
power $Q_i$ and the maximum transmit power of the relay is $P$.

Frequency division multiple access (FDMA) is used, so transmissions
of different users are orthogonal and interference-free. Without
loss of generality, we elaborate the transmission of User $i$'s
message on Channel $i$. We use the popular half-duplex two-step AF
relaying protocol. Let $s_i$ be the information symbol of User $i$.
It is normalized as $\mathbf{E}(|s_i|^2)=1$, where $\mathbf{E}$
stands for the average. In the first step, User $i$ transmits
$\sqrt{Q_i}s_i$. The signals received by the relay and Destination
$i$ are
\begin{equation}
    y_{iR}=\sqrt{Q_i}s_if_i+n_{iR}\quad\text{and}\quad y_{iD}=\sqrt{Q_i}s_ih_i+n_{iD},
\end{equation}
respectively, where $n_{iR}$ and $n_{iD}$ are the additive noises at
the relay and the destination in the first step, respectively. In
the second step, the relay amplifies its received signal and
forwards it to Destination $i$. Denote the power the relay uses to
help User $i$ as $P_i$. Since the relay has
perfect CSI, coherent power coefficient is used for better
performance \cite{add3,add4}. The signal received at Destination $i$
can be shown to be
\begin{equation}
    y_{Ri}=\sqrt{\frac{Q_iP_i}{Q_i|f_i|^2+1}}s_if_ig_i+\sqrt{\frac{P_i}{Q_i|f_i|^2+1}}g_in_{iR}+n_{RD},
\end{equation}
where $n_{RD}$ is the additive noise at the destination in the
second step. All noises are assumed to be i.i.d. additive circularly
symmetric complex Gaussian with zero-mean and unit-variance.

After maximum-ratio combining of both the direct path and the relay
path, the effective received signal-to-noise-ratio (SNR) of User
$i$'s transmission can be shown to be
\begin{equation}\label{eq-SNR}
\SNR_{i}=\frac{Q_iP_i|f_ig_i|^2}{P_i|g_i|^2+Q_i|f_i|^2+1}+Q_i|h_i|^2.
\end{equation}If  User $i$'s transmission is not helped by the relay and only the direct transmission is active, the
received SNR of User $i$'s transmission becomes
\begin{equation}
\SNR_{iD}=Q_i|h_i|^2.
\end{equation}

\subsection{Game Theoretical Model for Relay Pricing and Relay Power Allocation }
\label{sec-model-subSgame} In this paper, we focus on  the relay
power allocation among users and the relay power pricing. In early
relay network designs, many research focused on relay power
allocation  that optimizes the global network performance. In these
models, the relay has no gain and selflessly helps the users using
its own power; and the users are also selfless, who care about the
global network performance instead of their own benefits. In many
practical applications, however, the relay needs incentives for
cooperation. Also, there is a natural conflict among users, who want
to  obtain help from the relay to maximize their individual
benefits. With no payment for the relay power, each user wants as
much power as possible, which leads to undesirable situations, e.g.,
the total power demand from users exceeds the relay power budget.
This motivates the use of game theory to model the selfish behavior
of the users and the relay. Our goal is to find a fair power allocation
among the users and the optimal relay pricing strategy. We use the
Stackelberg game to model the interaction between the users and the
relay, and the bargaining game to model the relay power allocation
among the users, which, as explained in Section \ref{sec-background}, is
a natural fit.

We consider the relay as the leader of the Stackelberg game who sets
the price of its power in helping the users. The key point of the
relay game is for the relay to set the price to gain the maximum
revenue. The relay revenue, denoted as $u_R$, is the total payment
from the users. We use a simple pricing model by assuming that the
relay revenue is linear in the amount of power it sells, i.e.,
$u_R=\sum_{i=1}^N{\lambda P_i}$, where $\lambda$ is the normalized
unit price of the relay power and $P_i$ is the power the relay uses
to help User $i$.

We consider the users as followers of the Stackelberg game that
react in a rational way given the unit price of the relay power. The
bargaining game  is used to model the cooperative interaction among
users. That is, we assume that users make agreements to
cooperatively share the relay power. A key point of formulating the
users as selfish players in a bargaining game is to design the
utility function, which should reflect both the quality-of-service
and the payment-for-service of users. Its physical meaning can be
the benefits received by the users. In this paper, we seek to design
an appropriate utility function that is not only physically
meaningful, but also mathematically attractive to ensure
tractability and convergence.

We define the utility of User $i$ as
\begin{equation} \label{eq-userutility}
u_i\triangleq\frac{Q_iP_i|f_ig_i|^2}{P_i|g_i|^2+Q_i|f_i|^2+1}+Q_i|h_i|^2-\lambda
P_i,\quad i=1\cdots N,
\end{equation}
which, for a given network scenario, is a function of $P_i$, the
power the relay uses to help User $i$. The first two terms of
(\ref{eq-userutility}) correspond to the effective received SNR of
User $i$ given in (\ref{eq-SNR}) and represent the
quality-of-service of the user. It is directly related to the
performance of the communication, e.g., the achievable rate. The
last term $\lambda P_i$ represents the user's normalized cost in
purchasing the relay service.
If User $i$ does not buy any power from the relay and uses the
direct transmission only, i.e., $P_i=0$, its utility is the minimum
utility that User $i$ expects. Thus
\begin{equation}
u_{i,0}=Q_i|h_i|^2. \label{u-io}
\end{equation}

In the following section, we analyze the above Stackelberg game and
bargaining game models to find the optimal relay power pricing and a
fair power allocation among the users.

\section{Relay Power Allocation and Pricing Solutions}\label{sec-solution}
In this section, we solve the power allocation and pricing problems
jointly using the backward induction method \cite{*Owen}. That is,
we first solve the user game, i.e., the relay power allocation among
the users for a given price of the relay power, then solve the relay
game, i.e., the optimal price of the relay power, based on the
derived user bargaining strategy. The user game and the relay game
are formulated and analyzed in the following two subsections,
respectively.

\subsection{Relay Power Allocation Based on KSBS}
The user game is to find the relay power allocation among the users for
a given unit power price $\lambda$. We use the bargaining game in
cooperative game theory to find a fair power allocation.
Specifically, we look for the KSBS of the bargaining game, the
background of which is provided in Section
\ref{sec-backgroundsub-KSBS}.

We first calculate User $i$'s ideal utility $u^I_i$ of a given
$\lambda$. To maximize its utility, User $i$'s goal is
\begin{equation}\label{eq-nonusergame}
\max_{P_i}u_i\hspace{4mm} \ \mbox{s.t.}\ \ u_i\ge u_{i,0},\  0\le
P_i\le P .
\end{equation}
The first constraint in (\ref{eq-nonusergame}) ensures that User $i$
gets no less utility than $u_{i,0}$, which is its utility when it
receives no help from the relay, i.e., $P_i=0$. The second
constraint ensures that the power demand of User $i$ does not exceed
the total power budget $P$ of the relay. Given a relay power price,
this optimization problem can be solved analytically and the result
is given in Lemma \ref{lemma1}.
\newtheorem{lemma}{Lemma}

\begin{lemma}\label{lem-idealpower}
Define
\begin{equation}
b_i\triangleq\frac{Q_i|f_ig_i|^2}{Q_i|f_i|^2+1}. \label{def-b}
\end{equation}
Given the unit relay power price $\lambda$, the ideal power demand
of User $i$ that maximizes its utility $u_i$ in
(\ref{eq-userutility}) is
\begin{equation}\label{eq-idealpower}
P^I_i(\lambda)=\left\{
\begin{array}{ll}
0 & \text{if\ \ } \lambda\ge b_i\\
\frac{Q_i|f_i|^2}{\sqrt{b_i}}(\frac{1}{\sqrt{\lambda}}-\frac{1}{\sqrt{b_i}})&
\text{if\ \ } b_i>\lambda> b_i\left(\frac{b_iP}{Q_i|f_i|^2}+1\right)^{-2}\\
P & \text{if\ \ }\lambda\le
b_i\left(\frac{b_iP}{Q_i|f_i|^2}+1\right)^{-2}
\end{array}\right..
\end{equation}
The ideal utility of User $i$ is
\begin{equation}
u^I_i(\lambda)=\left\{
\begin{array}{ll}
u_{i,0} & \text{if }\lambda\ge b_i,\\
Q_i|f_i|^2(1-\sqrt{\lambda/b_i})^2+u_{i,0}& \text{if }b_i>\lambda>b_i\left(\frac{b_iP}{Q_i|f_i|^2}+1\right)^{-2}\\
\frac{b_iP}{(Q_i|f_i|^2)^{-1}b_iP+1}-\lambda P+u_{i,0} & \text{if }
\lambda\le  b_i\left(\frac{b_iP}{Q_i|f_i|^2}+1\right)^{-2}
\end{array}\right..
\label{ideal-u}
\end{equation}
\label{lemma1}
\end{lemma}

\begin{proof}
From (\ref{eq-userutility}), we have
\[\frac{\partial{u_i}}{\partial{P_i}}=b_i\left(\frac{b_iP_i}{Q_i|f_i|^2}+1\right)^{-2}-\lambda, \hspace{4mm}
\frac{\partial^2{u_i}}{\partial^2{P_i}}=\frac{-2(Q_i|f_i|^2)^{-1}b^2_i}{[(Q_i|f_i|^2)^{-1}b_iP_i+1]^3}.
\]
Thus $\frac{\partial^2{u_i}}{\partial^2{P_i}}<0$, which means that
$u_i$ is a concave function of $P_i$.

When $\lambda\ge b_i$, $\frac{\partial{u_i}}{\partial{P_i}}\leq 0$
for all $P_i \geq 0$ as $[(Q_i|f_i|^2)^{-1}b_iP_i+1]^2>1$. So $u_i$
is a non-increasing function of $P_i$ and its maximum is reached at
$P^I_i(\lambda)=0$. When $\lambda\le
b_i\left(\frac{b_iP}{Q_i|f_i|^2}+1\right)^{-2}$,
$\frac{\partial{u_i}}{\partial{P_i}}\geq 0$ for all $P_i\leq P$. So
$u_i$ is a non-decreasing function of $P_i$, and $P^I_i(\lambda)=P$
in this case. When $b_i>\lambda>
b_i\left(\frac{b_iP}{Q_i|f_i|^2}+1\right)^{-2}$, $u_i$ reaches its
maximum when $\frac{\partial{u_i}}{\partial{P_i}}=0$, i.e.,
$P_i=\frac{Q_i|f_i|^2}{\sqrt{b_i}}(\frac{1}{\sqrt{\lambda}}-\frac{1}{\sqrt{b_i}})$.
This proves the ideal power solution in (\ref{eq-idealpower}). Using
this solution and the equalities (\ref{eq-userutility}) and
(\ref{u-io}), we can obtain the ideal utility for User $i$ in
(\ref{ideal-u}).
\end{proof}

From Lemma \ref{lem-idealpower}, we see that $P_i^I(\lambda)$ is independent
of User $i$'s direct link $h_i$. Intuitively, this is because the
contribution of the direct link to User $i$'s receive SNR and
utility is fixed and keeps unchanged  for any amount of relay power
that User $i$ obtains.

Lemma \ref{lem-idealpower} also shows that when the price is too
high, Case 1 in (\ref{eq-idealpower}), User $i$ will not buy any
relay service. When the price is too low, Case 3 in
(\ref{eq-idealpower}), User $i$ wants to purchase all relay power to
maximize its utility. For the price range shown in Case 2 in
(\ref{eq-idealpower}), User $i$ asks for part of the relay power
that gives the ideal balance between its SNR and its payment to
maximize its utility. The ideal power demand of User $i$ depends not
only on the relay power price, but also on its power constraint
$Q_i$ and the quality of its local channels $f_i$ and $g_i$. The
$b_i$ defined in (\ref{def-b}), whose value depends on User $i$'s
condition only, is an important parameter. As shown in
(\ref{eq-idealpower}), it not only determines whether a relay asks
for the relay service but also affects how much power a user asks
for ideally. We can see $b_i$ as a quality measure for User $i$ to
some extent. For any two users, User $i$ and User $j$, assume that
$b_i>b_j$. We can see that if User $i$ is not allocated any relay
power, which happens when $b_i\le \lambda$, User $j$ will not be
allocated any relay power either because its $b_j$ is smaller. Also,
for a given price $\lambda$, increasing the $Q_i$ and $|f_i|^2$ of
User $i$, which increases $b_i$, results in higher or the same relay
power demand from User $i$, which is shown in the following lemma.
\begin{lemma}
Given a relay power price $\lambda$, $P_{i}^I(\lambda)$ is a
non-decreasing function of $Q_i$ and $|f_i|^2$. \label{P-b}
\end{lemma}

\begin{proof}
From (\ref{eq-idealpower}), we get, when $b_i>\lambda>
b_i\left(\frac{b_iP}{Q_i|f_i|^2}+1\right)^{-2}$,
\begin{equation*}
P^I_i(\lambda)=\frac{Q_i|f_i|^2}{b_i}\left(\sqrt{\frac{b_i}{\lambda}}-1\right)
=\frac{Q_i|f_i|^2+1}{|g_i|^2}\left\{\sqrt{\frac{|g_i|^2}{\lambda[1+1/(Q_i|f_i|^2)]}}-1\right\}.
\end{equation*}
For other two price ranges, when $\lambda\ge b_i$
$P^I_i(\lambda)=0$, and when $\lambda\le
b_i\left(\frac{b_iP}{Q_i|f_i|^2}+1\right)^{-2}$, $P^I_i(\lambda)=P$.
So, in all price ranges,
\begin{equation*}
P^I_i(\lambda)=\max\left[0,
\min\left(\frac{Q_i|f_i|^2+1}{|g_i|^2}\left\{\sqrt{\frac{|g_i|^2}{\lambda[1+1/(Q_i|f_i|^2)]}}-1\right\},P\right)\right].
\end{equation*}
$\max$ and $\min$ are
non-decreasing functions. For a given $\lambda$,
$\frac{Q_i|f_i|^2+1}{|g_i|^2}\left\{\sqrt{\frac{|g_i|^2}{\lambda[1+1/(Q_i|f_i|^2)]}}-1\right\}$
is also a non-decreasing function of $Q_i$ and $|f_i|^2$. So we conclude
that $P_{i}^I(\lambda)$ is a non-decreasing function of $Q_i$ and
$|f_i|^2$.
\end{proof}

To find the KSBS of the user bargaining game, without loss of
generality, we assume that the users are sorted in the descending
order of their $b_i$ values, that is
\begin{equation}
b_1\ge b_2\geq\cdots \geq b_N. \label{order-b}
\end{equation}
With the given price $\lambda$, for users satisfying $b_i\leq
\lambda$, as shown in Lemma \ref{lem-idealpower}, their ideal power
demand is 0 so they do not buy any power from the relay, thus do not
enter the game.

Let $L$ be the number of users satisfying $b_i > \lambda$. That is,
with the ordering in (\ref{order-b}), assume that
$b_L>\lambda>b_{L+1}$. The first $L$ users will participate in the
bargaining game and purchase the relay service. Given $\lambda$, to
find the KSBS-based power allocation of the $L$ users is equivalent
to solving the following optimization problem \cite{*Bargaining}:
\begin{equation}
    \max_{P_i}{k} \quad \mbox{s.t.}\quad \frac{\frac{b_iP_i}{(Q_i|f_i|^2)^{-1}b_iP_i+1}-\lambda P_i}{u^I_i-u_{i,0}}=k,
   \quad\sum^L_{i=1} P_i \leq P, \quad 0< P_i< Q_i|f_i|^2\left(\frac{1}{\lambda}-\frac{1}{b_i}\right),
\label{KSBS-prob}
\end{equation}
where $u_{i}^I$ and $u_{i,0}$, given in (\ref{ideal-u}) and
(\ref{u-io}) respectively, are the ideal and minimal utilities of
User $i$. The second constraint in (\ref{KSBS-prob}) is due to the
total power constraint of the relay, and the last constraint is to
ensure the feasibility of the solution and is derived from rewriting
$u_i> u_{i,0}$.

In the proof of Lemma \ref{lem-idealpower}, we have shown that $u_i$
is a concave function of $P_i$. Also, $u_i=u_{i,0}$ at $P_i=0$ and
$P_i=Q_i|f_i|^2\left(1/\lambda-1/b_i\right)$, and $u_i$ reaches its
maximum $u^I_{i}(\lambda)$ at $P_i=P_i^I(\lambda)$.  An example of
$u_i$ as a function of $P_i$ is given in Figure
\ref{fig:ui-concavity}. It can be shown from the utility definition
in (\ref{eq-userutility}) that for each
$u\in(u_{i,0},u^I_{i}(\lambda))$, there are two possible choices of
$P_i$ in the range
$\left(0,Q_i|f_i|^2\left(1/\lambda-1/b_i\right)\right)$ that satisfy
$u_i(P_i)=u$: one in the range $\left(0,P_i^I(\lambda) \right]$ and
the other in the range
$\left[P_i^I(\lambda),Q_i|f_i|^2\left(1/\lambda-1/b_i\right)\right)$.
Thus we can shrink the feasible region of $P_i$ from
$\left(0,Q_i|f_i|^2\left(1/\lambda-1/b_i\right)\right)$ to either
one of the smaller regions.
\begin{figure}[!ht]
  \centering
   \psfrag{high}{$Q_i|f_i|^2\left(\frac{1}{\lambda}-\frac{1}{b_i}\right)$}
\psfrag{midle}{$P_i^I(\lambda)$}
  \psfrag{Ui}{$u_i$}
  \psfrag{Ui0}{$u_{i,0}$}
  \psfrag{UiI}{$u^I_{i}(\lambda)$}
  \psfrag{Pi}{$P_i$}
  \psfrag{0}{$0$}
  \includegraphics[width=.5\textwidth]{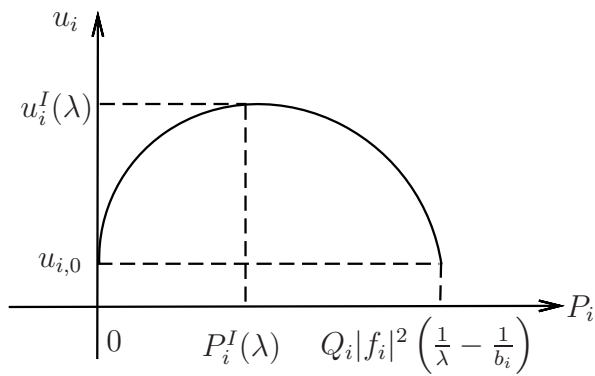}\\
  \caption{Concavity of the utility function.}
  \label{fig:ui-concavity}
\end{figure}
We choose the first region for two reasons. First, for the same
$u_i$ value, this choice results in a smaller $P_i$ than choosing
the second region, and the users prefer to buy less power to gain
the same utility. Second, smaller power consumption for each user
saves relay power, so more users can be helped. Thus,
(\ref{KSBS-prob}) becomes
\begin{eqnarray}
    \max_{P_i}{k}\quad \mbox{s.t.}\quad \frac{\frac{b_iP_i}{(Q_i|f_i|^2)^{-1}b_iP_i+1}-\lambda P_i}{u_i^I-u_{i,0}}= k, \quad \sum^L_{i=1} P_i \leq P,\quad 0< P_i\leq P_i^I(\lambda).
\label{KSBS-prob-equi10}
\end{eqnarray}
To solve this optimization problem, we prove the following lemma.
\begin{lemma}
The relay power allocation problem in (\ref{KSBS-prob-equi10}) is
equivalent to the following max-min problem:
\begin{eqnarray}
\max_{P_i}\min_i\Bigg\{
\frac{\frac{b_iP_i}{(Q_i|f_i|^2)^{-1}b_iP_i+1}-\lambda
P_i}{u_i^I-u_{i,0}}\Bigg\}, \quad \mbox{s.t.}\quad\sum^L_{i=1} P_i
\leq P,\quad 0< P_i\leq P_i^I(\lambda). \label{coop-user1}
   \end{eqnarray}
\label{lemma-equi}
\end{lemma}
\begin{proof}
First we use the notation $\psi_i(P_i)\triangleq
\frac{\frac{b_iP_i}{(Q_i|f_i|^2)^{-1}b_iP_i+1}-\lambda
P_i}{u_i^I-u_{i,0}}$. To prove this lemma, it is sufficient to show
that the power allocation solution in (\ref{coop-user1}), denoted as
$(P_1^*,\cdots,P_L^*)$, satisfies
$\psi_1(P_1^*)=\cdots=\psi_L(P_L^*)$. We prove this by
contradiction. Without loss of generality, suppose that
$\psi_1(P_1^*)<\psi_2(P_2^*)<\min\{\psi_3(P_3^*),
\cdots,\psi_L(P_L^*)\}$. Thus,
$\max_{P_i}\min_i\psi_i(P_i^*)=\psi_1(P_1^*)$. Since
$\psi_1(P_1),\psi_2(P_2)$ are increasing and continuous functions of
$P_1,P_2$ in the feasible region given in (\ref{coop-user1}), there
exists a small enough positive $\epsilon$ such that $P_1^*+\epsilon,
P_2^*-\epsilon$ are still in the feasible region and
\[\psi_1(P_1^*)<\psi_1(P_1^*+\epsilon)<\psi_2(P_2^*-\epsilon)<\psi_2(P_2^*).
\]
The new power allocation
$(P_1^*+\epsilon,P_2^*-\epsilon,P_3^*,\cdots,P_L^*)$ satisfies all
power constraints in (\ref{coop-user1}). Its max-min value is
$\psi_1(P_1^*+\epsilon)$ which is larger than the max-min value of
the solution $(P_1^*,\cdots,P_L^*)$. This contradicts the assumption
that $(P_1^*,\cdots,P_L^*)$ is optimal, thus completes the proof.
\end{proof}

(\ref{coop-user1}) is a convex optimization problem and can be
solved efficiently using standard convex optimization techniques \cite{*boyd}.
We call the solution of (\ref{coop-user1}) the KSBS-based power
allocation. Recall that in (\ref{coop-user1}), only the $L$ users
whose $b_i$'s are larger than the relay price $\lambda$ participate
in the game. The remaining $N-L$ users request no relay power.

In the game theoretical model in (\ref{coop-user1}), the power
constraint at the relay is taken into consideration. For any relay
price $\lambda$, (\ref{coop-user1}) will result in a feasible power
allocation among users, i.e., the total power demanded by the users
does not exceed the relay power constraint. Without the game
theoretical model, if, for example, for a given price, the users
request their ideal relay powers to maximize their individual
utilities, it may happen that the total power demand of the users
exceeds the relay power constraint, and the relay cannot make a
satisfactory power allocation. With the proposed KSBS-based relay
power allocation, when the sum of the ideal power demands of all
users does not exceed the relay power constraint, the users will be
allocated their ideal powers, in which case, $k$ in
(\ref{KSBS-prob-equi10}) reaches its maximum 1; when the sum of the
ideal power demands of all users exceeds the relay power constraint,
the proposed KSBS-based power allocation will allocate all relay
power to the users fairly. This is shown in the following lemma.
\begin{lemma}\label{lem-sumpower}
For a fixed $\lambda$, let the ideal power allocation of User $i$ be
$P^I_i(\lambda)$, which is given in (\ref{eq-idealpower}); and let
the KSBS-based power allocation be $P_i^K(\lambda)$ (K stands for
KSBS). When $\sum_{i=1}^L P^I_i(\lambda)\leq P$, we have
$P_i^K(\lambda)=P_i^I(\lambda)$; when $\sum_{i=1}^L
P^I_i(\lambda)>P$, we have $\sum_{i=1}^L P^K_i(\lambda)=P$.
\end{lemma}
\begin{proof} Again, we use the notation $\psi_i(P_i)\triangleq \frac{\frac{b_iP_i}{(Q_i|f_i|^2)^{-1}b_iP_i+1}-\lambda
P_i}{u_i^I-u_{i,0}}$. With the new feasible region of $P_i$ in
(\ref{coop-user1}), $\psi_i(P_i)$'s are increasing functions and
reach their maximum 1 when $P_i=P_i^I(\lambda)$. Thus $k \in [0,1]$
and achieves the maximum $k=1$ if and only if
$\sum_{i=1}^L{P^I_i(\lambda)} \leq P$, that is, when users can reach
their ideal utility with a feasible relay power. In this case,
$P_i^K(\lambda)=P_i^I(\lambda)$.

If $\sum_{i=1}^L{P^I_i(\lambda)}>P$, not all users can reach their
ideal utilities and thus $k<1$. From the equivalent form
(\ref{KSBS-prob-equi10}), actually no user can reach its ideal
utility. That is, $P_i^K(\lambda)<P_i^I(\lambda)$. Suppose that
$\sum_{i=1}^L{{P^K_i(\lambda)}}<P$. Define
\[\epsilon\triangleq\min_i\left\{\frac{P-\sum_{i=1}^L{{P^K_i(\lambda)}}}{L}, P^I_1(\lambda)-P^K_1(\lambda),\cdots,P^I_L(\lambda)-P^K_L(\lambda)\right\}.\] $\epsilon$ is a positive number. Now consider the power allocation $\tilde{P}_i(\lambda)\triangleq P_i^K(\lambda)+\epsilon$.
First, this new power allocation satisfies all power constraints due
to its construction. Also, as $\psi_i$'s are increasing functions,
the new power allocation results in a higher minimum value, that is
$\min_i\psi_i(\tilde{P}_i(\lambda))>\min_i\psi_i(P_i^K(\lambda))$,
which contradicts the assumption that $P_i^K(\lambda)$ is optimal.
This completes the proof.
\end{proof}

\subsection{Optimal Relay Power Price}
\label{sec-relay} Now we investigate the relay pricing problem. The
price of the relay power is crucial to the relay revenue and the
relay power allocation among the users. If the relay sets the price too
high, no user will buy any power, and the relay revenue will be
zero. If the relay sets the price too low, all users will ask for as
much power as possible; and even though all relay power can be sold,
the relay revenue will not be maximized.

With the unit price of the relay power $\lambda$, from Section III,
and by using the KSBS-based relay power allocation in Section IV.A,
the revenue of the relay is $\sum_{i=1}^N{\lambda P^K_i(\lambda)}$,
where $P^K_i(\lambda)$ is the relay power allocated to User $i$
based on the KSBS for the given price $\lambda$. The relay pricing
problem can be formulated as:
\begin{equation}
    \max_{\lambda}\sum_{i=1}^N{\lambda P^K_i(\lambda)}.
\label{relay-prob-coop}
\end{equation}
Note that the relay power constraint $\sum_{i=1}^NP_i^K(\lambda)\le
P$ is always guaranteed by the KSBS-based power allocation, thus
needs not to appear explicitly in the relay revenue maximization.

To solve the relay pricing problem, we first prove the following
lemma.
\begin{lemma}
The optimal price is inside the interval $[b_{lb},b_1)$, where
$b_{lb}$ satisfies the following equation:
\begin{equation}
\phi(b_{lb})\triangleq
\sum_{i=1}^N\max{\left\{0,\frac{Q_i|f_i|^2}{\sqrt{b_i}}\left(\frac{1}{\sqrt{b_{lb}}}-\frac{1}{\sqrt{b_i}}\right)\right\}}=P
\label{b_{lb}}
\end{equation}
and
\begin{equation}
b_{lb}\ge\max_i\left\{b_{i}\left(\frac{b_iP}{Q_i|f_i|^2}+1\right)^{-2}\right\}
\label{temp2}
\end{equation}
\label{lemma-range}
\end{lemma}

\begin{proof}
First we can see that $\phi(b_{lb})$ monotonically decreases from
$\infty$ to $0$ as $b_{lb}$ increases from 0 to $b_1$. Thus, the
equation (\ref{b_{lb}}) has a unique positive solution inside
$(0,b_1)$.

Then we prove (\ref{temp2}) by contradiction. Assume that
$b_{lb}<b_{1}\left(\frac{b_1P}{Q_1|f_1|^2}+1\right)^{-2}$. Thus,
\[
\phi(b_{lb})\ge \max\left\{0,\frac{Q_1|f_1|^2}
{\sqrt{b_1}}\left(\frac{1}{\sqrt{b_{lb}}}-\frac{1}{\sqrt{b_1}}\right)\right\}
>P,
\]
which conflicts (\ref{b_{lb}}). So $b_{lb}\geq
b_{1}\left(\frac{b_1P}{Q_1|f_1|^2}+1\right)^{-2}$. Similarly, we can
show that $b_{lb}\geq
b_{i}\left(\frac{b_iP}{Q_i|f_i|^2}+1\right)^{-2}$ for
$i=2,\cdots,N$. Thus (\ref{temp2}) is proved.

Now we show that the optimal price is no less than $b_{lb}$. Using
the result in (\ref{temp2}) and from (\ref{eq-idealpower}), when the
relay power price is $b_{lb}$, i.e., $\lambda=b_{lb}$, we have
\begin{equation}
\sum_{i=1}^N{P^I_i(b_{lb})}=\phi(b_{lb})=P. \label{temp1}
\end{equation}
Also from (\ref{eq-idealpower}), $P^I_i(\lambda)$ is a continuous
and non-increasing function of $\lambda$. So
$\sum_{i=1}^N{P^I_i(\lambda)}$ is a continuous and non-increasing
function of $\lambda$. Inside the price range $[0,b_{lb}]$, i.e.,
$\lambda< b_{lb}$, we have $\sum^L_{i=1}{P^I_i(\lambda)}\geq P$
based on (\ref{temp1}). With the KSBS-based power allocation,
according to Lemma \ref{lem-sumpower}, all power of the relay will
be allocated to the users , i.e., $\sum^L_{i=1}{P^K_i(\lambda)}=P$.
The relay revenue maximization when the price is with $[0, b_{lb}]$
becomes:
\begin{equation}
\max_{0\leq\lambda\leq
b_{lb}}{\lambda\sum^L_{i=1}{P^K_i(\lambda)}}=b_{lb} P,
\end{equation}
which is reached at $\lambda=b_{lb}$. So the optimal price in the
range $[0,b_{lb}]$ is $b_{lb}$, which shows that the optimal price
is no less than $b_{lb}$.

To prove the upper bound on the relay price, note that when $\lambda
\geq b_1$, from (\ref{eq-idealpower}), $P_i^I(\lambda)=0$ for all
$i$, i.e., no user will buy any power from the relay and the relay
revenue will be $0$. So any price in the range
$\left[b_1,+\infty\right)$ is not optimal for the relay revenue, and
the optimal price must be in the range $[b_{lb}, b_1)$.
\end{proof}

The value of $b_{lb}$ can be obtained by solving the equation in
(\ref{b_{lb}}). This is a generalized waterfilling problem
\cite{*waterfilling}, where $1/\sqrt{\lambda}$ is the water-level,
$1/\sqrt{b_i}$ is the ground level of User $i$, and
$Q_i|f_i|^2/\sqrt{b_i}$ are the weights that can be visually
interpreted as the width of each patch. In this paper, we can find
the value of $b_{lb}$ analytically. Notice that $\phi(b_{lb})$ is a
decreasing function of $b_{lb}$ and $b_i$'s are in non-increasing
order. We can first find the $M$ such that $\phi(b_M)<P$ and
$\phi(b_{M+1})>P$. Thus, $b_{lb}\in[b_M,b_{M+1}]$. Within this
interval,
$\phi(b_{lb})=\sum_{i=1}^M\frac{Q_i|f_i|^2}{\sqrt{b_i}}\left(\frac{1}{\sqrt{b_{lb}}}-\frac{1}{\sqrt{b_i}}\right)
=\left(\sum_{i=1}^M\frac{Q_i|f_i|^2}{\sqrt{b_i}}\right)\frac{1}{\sqrt{b_{lb}}}-\left(\sum_{i=1}^M\frac{Q_i|f_i|^2}{b_i}\right)$.
Thus, from $\phi(b_{lb})=P$, we have
\begin{equation}
b_{lb}=\left(\sum_{i=1}^M\frac{Q_i|f_i|^2}{\sqrt{b_i}}\right)^{2}\left(P+\sum_{i=1}^M\frac{Q_i|f_i|^2}{b_i}\right)^{-2}.
\label{b-lb}
\end{equation}



In what follows, we solve the optimal relay power price
analytically. First, several notation are introduced. Recall the
ordering of the users based on their $b_i$ values in (\ref{order-b}) and
$M$ is the index such that $b_M\ge b_{lb}\ge b_{M+1}$. That is, the
$b_i$'s of the first $M$ users are no less than $b_{lb}$, while the
$b_i$'s of the remaining users are no larger than $b_{lb}$.
We have shown in Lemma \ref{lemma-range} that only the price range
$[b_{lb},b_1)$ needs to be considered for the optimal price. Define $\gamma_i\triangleq b_i$ for $i=1,\cdots,M$ and
$\gamma_{M+1}\triangleq b_{lb}$. Further define $\Gamma_1\triangleq
[\gamma_2,\gamma_1)$ and $\Gamma_i\triangleq
[\gamma_{i+1},\gamma_i]$ for $i=,2\cdots,M$. We thus can have divided
the price range $[b_{lb},b_1)$ into the following $M$ intervals:
\setlength{\arraycolsep}{0pt}
\begin{eqnarray}
[b_{lb},b_1) &=& [b_{lb},b_M]\cup[b_{M}, b_{M-1}]\cup\cdots\cup[b_{3},b_{2}]\cup[b_{2}, b_1) \nonumber\\
&\triangleq&\Gamma_M\cup\Gamma_{M-1}\cdots\cup\Gamma_2\cup\Gamma_1.
\label{int-decomp}
\end{eqnarray}
\setlength{\arraycolsep}{5pt}

Inside
the price range $[b_{lb},b_1)$, because $\sum_{i=1}^NP_i^I(\lambda)$ is a
non-increasing function of $\lambda$ and (\ref{b_{lb}}), we have
$\sum_{i=1}^NP_i^I(\lambda)\le P$. Thus,
$P_i^K(\lambda)=P_i^I(\lambda)$ from Lemma \ref{lem-sumpower}. We
can thus rewrite the price optimization problem in
(\ref{relay-prob-coop}) into
\begin{equation}
    \max_{i=1,2,\cdots M}{\max_{\lambda\in\Gamma_i}{\sum^i_{j=1}{\lambda P_j^I(\lambda)}}}.
\label{rewrite}
\end{equation}
In (\ref{rewrite}), we have decomposed the optimization problem into
$M$ subproblems, where the $i$th subproblem is to find the optimal
price within the range $\Gamma_i$ where User 1 to $i$ purchase
non-zero power from the relay:
\begin{equation}
\mbox{Sub-problem $i$}:\hspace{4mm}
\max_{\lambda\in\Gamma_i}{\sum^i_{j=1}{\lambda P_j^I(\lambda)}}.
\label{sub-prob}
\end{equation}
The following proposition is proved to solve the sub-problem.
\newtheorem{proposition}{Proposition}
\begin{proposition}\label{prop-subprob}
For $i=1,2,\cdots,M$, define
\begin{equation}
c_i\triangleq\left(\frac{\sum^i_{j=1}Q_j|f_j|^2/\sqrt{b_j}}{2\sum^i_{j=1}Q_j|f_j|^2/b_j}\right)^2.
\label{def-c}
\end{equation}
 The solution to (\ref{sub-prob}) is
\begin{equation}
\quad \lambda_i\triangleq\left\{
\begin{array}{ll}
\gamma_{i+1} & \text{if }\ c_i<\gamma_{i+1},\\
\gamma_i& \text{if }\ c_i>\gamma_i,\\
c_i& \text{if }\ \gamma_{i+1}\leq c_i\leq \gamma_i.
\end{array}\right.
\label{sub-solution}
\end{equation}
\end{proposition}

\begin{proof}
When $\lambda\in\Gamma_{i}$, for $1\leq j\leq i$, from
(\ref{eq-idealpower}), User $j$ will ask for
$\frac{Q_i|f_i|^2}{\sqrt{b_i}}\left(\frac{1}{\sqrt{\lambda}}-\frac{1}{\sqrt{b_i}}\right)$
units of power, and User $(i+1)$ to User $M$ will ask for zero relay
power. Subproblem (\ref{sub-prob}) can be rewritten as
\begin{equation}
\max_{\lambda\in\Gamma_i}
\left\{\lambda\sum^i_{j=1}{\frac{Q_i|f_i|^2}{\sqrt{b_i}}\left(\frac{1}{\sqrt{\lambda}}-\frac{1}{\sqrt{b_i}}\right)}\right\}
=\max_{\lambda\in\Gamma_i} \phi_{R,i}(\lambda),
    \label{sub_uR}
\end{equation}
where
$\phi_{R,i}(\lambda)\triangleq\left(\sum^i_{j=1}\frac{Q_j|f_j|^2}{\sqrt{b_j}}\right)\sqrt{\lambda}
-\left(\sum^i_{j=1}\frac{Q_j|f_j|^2}{b_j}\right) \lambda$. In
(\ref{sub_uR}), $\phi_{R,i}(\lambda)$ is the relay revenue given the
price $\lambda \in \Gamma_i$. It can be shown through
straightforward calculation that $d\phi_{R,i}(\lambda)/d\lambda=0$
when $\lambda=c_i$, defined in Proposition \ref{prop-subprob}, and
$\frac{d^2\phi_{R,i}(\lambda)}{d\lambda^2}<0$, when
$\lambda\in\Gamma_i$. Therefore, if $ c_i>\gamma_i$,
$\phi_{R,i}(\lambda)$ reaches its maximum at $\gamma_i$; if
$c_i<\gamma_{i+1}$, it reaches its maximum at $\gamma_{i+1}$; and if
$\gamma_{i+1}\leq c_i \leq \gamma_i$, it reaches its maximum at
$c_i$.
\end{proof}

With the subproblems solved, we are ready to find the optimal relay
power price. The result is given in the following theorem.
\newtheorem{theorem}{Theorem}
\begin{theorem}\label{theorem-pricing}
The optimal relay power price, denoted as $\lambda^*$, is
\begin{equation}
  \lambda^*= \arg \max_{\lambda_i}\left\{ \left(\sum^i_{j=1}\frac{Q_j|f_j|^2}{\sqrt{b_j}}\right)\sqrt{\lambda_i}-\sum^i_{j=1}\frac{Q_j|f_j|^2}{b_j}\lambda_i\right\},
\label{opt-price}
\end{equation}
where $\lambda_i$ is defined in Proposition \ref{prop-subprob}.
\end{theorem}
\begin{proof}
This is a natural result of Proposition \ref{prop-subprob} and
(\ref{rewrite}).
\end{proof}

With Theorem \ref{theorem-pricing}, we can find the optimal price
for the relay power by solving the $M$ subproblems  in
(\ref{rewrite}) analytically using Proposition \ref{prop-subprob},
then find the optimal price among the $M$ sub-problem solutions that
results in the maximum relay revenue. This is written as Algorithm
\ref{alg-pricing}. We can see that, after ordering (whose average
complexity is $N\log{N}$), its complexity is linear in the number of
users in the network.

\begin{algorithm}
\caption{Optimal relay power price.} \label{alg-pricing}
\begin{algorithmic}[1]
\STATE Calculate $b_i$'s using (\ref{def-b}). Order the $N$ users such
that
 $b_1\geq b_2\geq \cdots \geq b_N$.
\STATE  Find $M$ then $b_{lb}$ using (\ref{b-lb}). \STATE Initialize
$\gamma_i$: $\gamma_i=b_i$ for $i=1,\cdots,M$ and
$\gamma_{M+1}=b_{lb}$. \STATE Calculate $c_i$'s for $i=1,\cdots,M$
using (\ref{def-c}). \STATE For $i=1,\cdots,M$, find $\lambda_i$
using (\ref{sub-solution}). \STATE Find the optimal price
$\lambda^*$ using (\ref{opt-price}).
\end{algorithmic}
\end{algorithm}

Previously, we have shown that $b_i$ is an important factor for the
ideal relay power. Here we can see that it is also important for the
optimal relay price. We prove the following lemma, which further
reflects the importance of $b_i$.
\begin{lemma}
If $b_1<4b_{lb}$, the optimal price for the relay is $b_{lb}$.
\label{remark-optimal-price}
\end{lemma}

\begin{proof}
First recall that $b_1\ge \cdots\ge b_{M-1}\ge b_{M}\ge b_{lb}$.
When $b_1<4b_{lb}$, for $i=1,\cdots,M$ and $j=1,\cdots,i$, we have
$b_j\le b_1<4\gamma_{M+1}\le4\gamma_{i+1}$. Therefore,
\begin{equation*}
 \frac{Q_j |f_j|^2}{\sqrt{4 \gamma_{i+1}}} < \frac{Q_j |f_j|^2}{\sqrt{b_j}} \Leftrightarrow
 \frac{Q_j|f_j|^2}{\sqrt{b_j}}<\frac{2Q_j|f_j|^2\sqrt{\gamma_{i+1}}}{b_j},
\end{equation*}
and
\begin{equation*}
 \sqrt{c_i}=\frac{\sum^i_{j=1}Q_j|f_j|^2/\sqrt{b_j}}{2\sum^i_{j=1}Q_j|f_j|^2/b_j}
 <\frac{2\sqrt{\gamma_{i+1}}\sum^i_{j=1}Q_j|f_j|^2/b_j}{2\sum^i_{j=1}Q_j|f_j|^2/b_j}=\sqrt{\gamma_{i+1}}
\end{equation*}
for $i=1,\cdots,M$. From Proposition 1, within the range $\Gamma_i$
where $i=1, \cdots, M$, the optimal price is $\gamma_{i+1}$, the
lower bound of $\Gamma_i$. So the optimal price in the range
$[\gamma_{M+1},\gamma_1)$ is $\gamma_{M+1}$, which is $b_{lb}$.
\end{proof}
Lemma \ref{remark-optimal-price} says that when the difference
between $b_1$ and $b_{lb}$ is small, that is, the conditions of the
users are not too separate apart, the relay should set its price to
be low so all users can gain some benefits. On the contrary, when
some users have a much higher $b_i$ than others, the price will be
higher than $b_{lb}$ and those users with lower $b_i$'s may not
purchase the relay service because the price is too high compared to
the SNR gain they may receive.

\section{Discussion}\label{sec-discussion}
In this section, we discuss possible implementation of the proposed
relay power allocation and pricing solutions, properties of the
power allocation solution, and applications of the proposed
solutions to some special network scenarios.

The first to discuss is the implementation of the proposed relay
power allocation and power pricing solutions. In practical network
applications, to use the proposed scheme, the relay, which is the
service provider and has perfect and global channel knowledge, first
finds the optimal price of the relay power using Algorithm
\ref{alg-pricing}. With this optimal price, the relay then finds the
KSBS-based solution for the relay power allocation problem given in
(\ref{coop-user1}). With this implementation, we actually assume
that the relay is trustworthy. All users believe that the relay
will not change the parameter values (e.g., the CSI); and the relay
uses the above procedure to set the price and determine the
KSBS-based power allocation, and follows the results to help all
users in their transmission. For the relay to know the channel
gains from the users to itself, training and channel estimation
should be performed at the relay. For the relay to know the channel
gains from itself to the destinations, feedback from the
destinations to the relay are required. The proposed algorithm is a
centralized one instead of distributed.

Now we discuss properties of the KSBS-based power allocation.
When the relay sets its price to be the optimal, from the analysis
in Section \ref{sec-relay}, all users will be allocated their ideal
relay powers, $P_i^I(\lambda)$, and the individual utilities of the
users are maximized. This is the ideal case and requires the relay
to have perfect CSI. However, in reality, CSI at the relay is
subject to error and delay, in which case, the relay may
set a price different to the optimal one. Sometimes, the relay may
want to set its price different to the optimal one due to other
reasons such as marketing considerations. Our bargaining game model
and KSBS-based power allocation is robust to the relay price
fluctuation in the sense that a ``fair" relay power allocation among
the users can still be made. Specifically, if the relay power price is
set to be higher than or equal to $b_{lb}$, defined in
(\ref{b_{lb}}), with the KSBS-based power allocation, each user gets
its ideal power demand; if the relay power price is set to be lower
than $b_{lb}$, no user can get its ideal relay power but the relay
power will be fairly allocated to the users such that the utility
losses of the  users are the same in the logarithmic scale; and all
relay power will allocated.

Finally, we discuss the application of the proposed solution to two
special network scenarios. One network application is the
multi-user, single-relay, and single-destination network, also
addressed as multi-access relay network (MARN) in some papers
\cite{*MARN1,*MARN3,*LLB1,*LLB2,*Tairan,*Sha,*Guftaar}. The proposed
scheme can be directly applied to MARNs by setting $g_1=\cdots=g_N$
in the network formulation. From Lemma \ref{P-b}, $P^I_i(\lambda)$
is a non-decreasing functions of $Q_i$ and $f_i$. Thus, with the
relay to destination channel the same for all users, users with
better user-relay channels or higher transmit powers will be
allocated more relay power. Another popular network scenario is the
multi-user single-relay network with no direct links. Our solutions
again can be applied straightforwardly as the solutions are
independent of the direct link.

\section{Simulation Results}\label{sec-simulation}
In this section, we show the simulated performance of the proposed
relay power allocation and pricing solutions, and compare them with
the sum-rate-optimal power allocation and the even power allocation.
Sum-rate-optimal power allocation solution is the relay power
allocation among the users that maximizes the network sum-rate. For the even power allocation, the relay
allocates $1/N$ of its total power to each of the $N$ user, and each
user decides how much power to buy from the relay to maximize its
utility. That is, the relay power allocated to User $i$ is
$\min\{P_i^I(\lambda), P/N\}$. Two channel models are considered:
the Rayleigh flat-fading channel and the static channel with
path-loss only.
\subsection{Network with Rayleigh Flat-Fading Channels}

\begin{figure}
  \centering
  \includegraphics[width=0.6\textwidth]{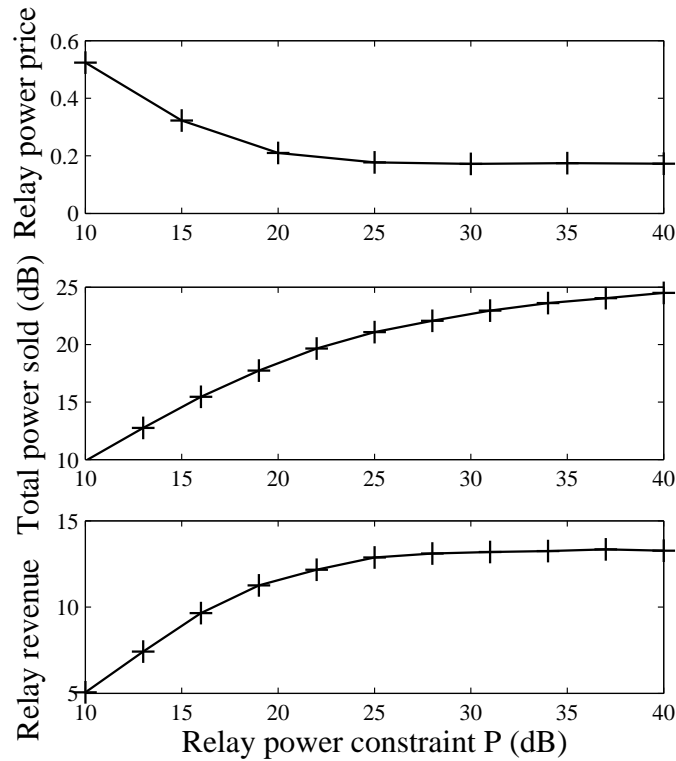}\\
  \caption{Three-user relay network with Rayleigh fading channels and different relay power constraints.}\label{simu:1a}
\end{figure}

In the first numerical experiment, the channels are modeled as
i.i.d.~Rayleigh flat-fading, i.e., $f_i, h_i$, and $g_i$ are
generated as i.i.d.~random variables following the distribution
$\mathcal{CN}(0,1)$. We consider a network with three users. The
transmit powers of the users are set to be $10$ dB. The simulation
results follow the same trend for other values of user powers.

We first investigate the network performance when the relay power
ranges from $10$ dB to $40$ dB. This corresponds to the scenario
where the total user demand is fixed while the relay power supply
increases. We set the relay power price to be the optimal according
to Theorem \ref{prop-subprob}. Figure \ref{simu:1a} shows the
optimal relay power price, the relay power actually sold, and  the
maximum relay revenue under different relay power constraints. We
can see that when the relay has more power to sell, the optimal
relay power price is lower, more relay power is sold, and the relay
receives more revenue. This complies with one of the laws of supply
and demand \cite{*Besanko}, which says that if supply increases and
demand remains unchanged, then it leads to lower equilibrium price
and higher quantity.

Figure \ref{simu:1b} compares the network sum-rate and fairness of
the proposed KSBS-based power allocation with those of the
sum-rate-optimal power allocation and the even power allocation. We
set the relay power price to be the optimal according to Theorem
\ref{prop-subprob}. It can be seen that for the system sum-rate, the
difference between our algorithm and the sum-rate-optimal solutions
is within $3.5\%$, while it is within $13\%$ between the
sum-rate-optimal and the even power solutions. The proposed solution
is about $5$ dB superior to the even power allocation. To quantify
the fairness of different allocation schemes, we use the average
value of the normalized difference:
$[\max_i(r_i)-\min_i(r_i)]/\max_i(r_i)$, where $r_i$ is the
achievable rate of User $i$. A smaller difference indicates a fairer
solution. We can see that our solution achieves similar fairness to
the even power solution and is fairer than the sum-rate-optimal one.
\begin{figure}
  \centering
  \includegraphics[width=0.6\textwidth]{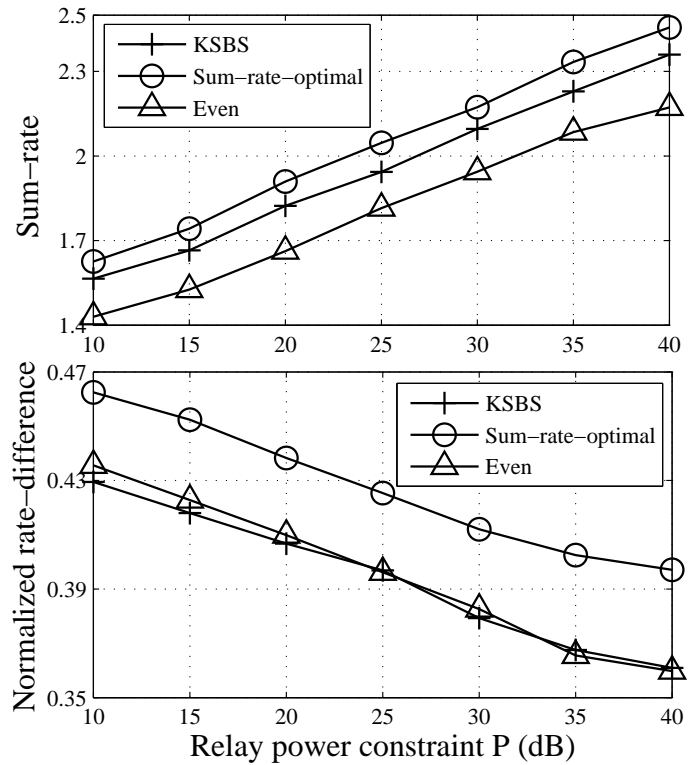}\\
  \caption{System sum-rate and fairness of a three-user relay network with Rayleigh-fading channels.}\label{simu:1b}
\end{figure}

\begin{figure}
  \centering
  \includegraphics[width=0.6\textwidth]{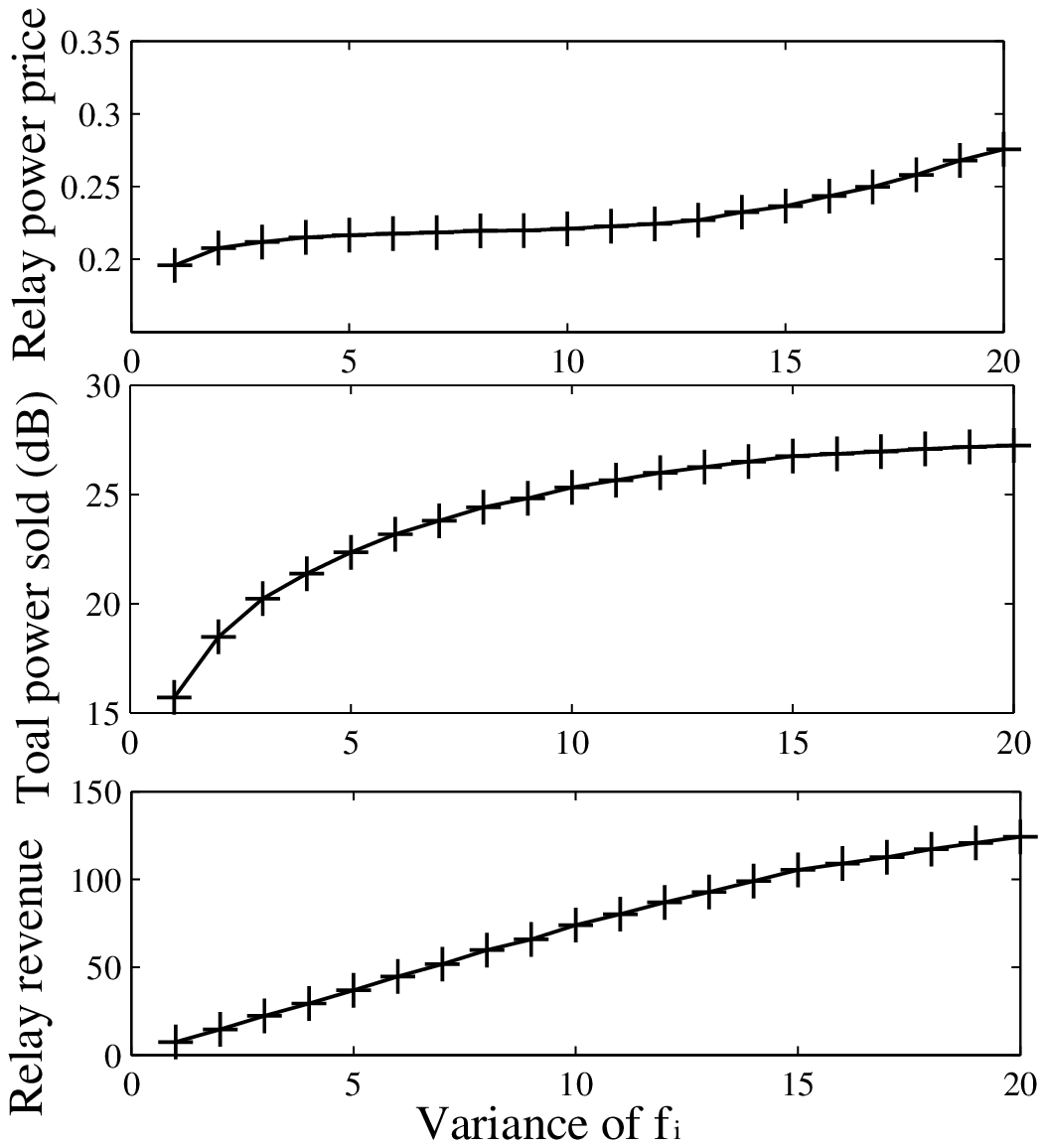}\\
  \label{fig-NWModel}
  \caption{Three-user relay network with Rayleigh fading channels and different variances of $f_i$.}
\label{simu:2a}
\end{figure}
Next, we examine the trend of the optimal relay price with an
increasing demand. From Lemma \ref{P-b}, $P^I_i(\lambda)$ is a
non-decreasing function of $|f_i|^2$. So, we can use an increasing
$|f_i|^2$ to simulate the increasing user demand. In this numerical
experiment, we again consider a three-user network and model all
channels as independent circularly symmetric complex Gaussian random
variables with zero-mean, that is, they are independent Rayleigh
flat-fading channels. The variances of all $g_i$'s and $h_i$'s are
1, while the variance  of all $f_i$'s ranges from 1 to 20. A larger
variance means a higher average value of $|f_i|^2$, which on average
means a higher power demand from the users. The transmit power of
the users is set to be $10$ dB and relay power is set to be $20$ dB.
Figure \ref{simu:2a} shows the optimal relay power price, the actual
relay power sold, and the maximum relay revenue with different
variances of $f_i$. We can see that as the variance of $f_i$
increases, the optimal relay price increases, more relay power is
sold, and the maximum relay revenue increases. This fits one of the
laws of supply and demand, which says, if the supply is unchanged
and demand increases, it leads to higher equilibrium price and
quantity.

\begin{figure}
  \centering
  \includegraphics[width=0.6\textwidth]{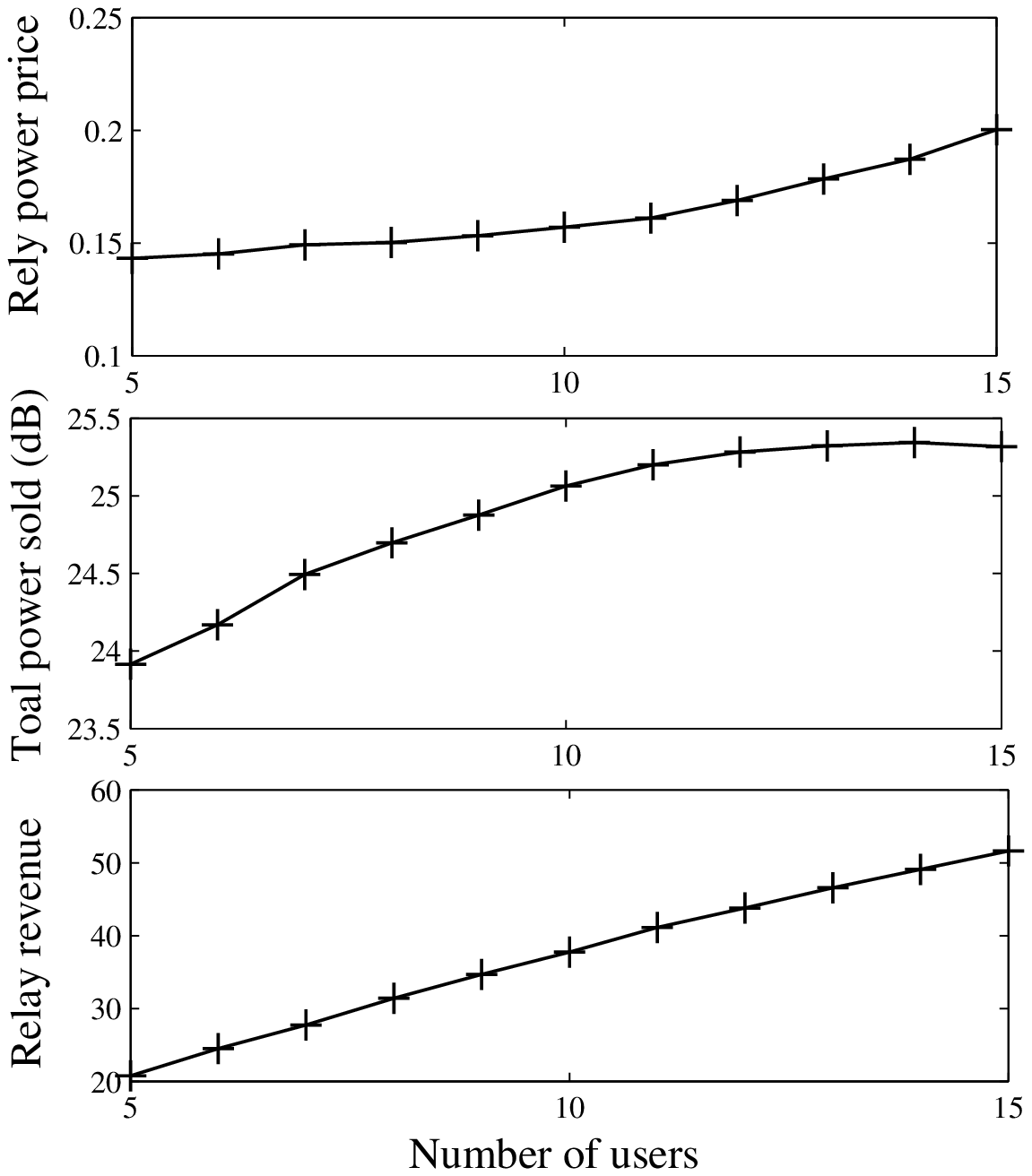}\\
  \label{fig-NWModel}
  \caption{Three-user relay network with Rayleigh fading channels and different numbers of users.}\label{simu:2b}
\end{figure}

In the third numerical experiment, we examine the relationship
between the optimal relay price and the number of users. We assume
that the relay power is fixed to be $20$ dB. The user power is fixed
as $10$ dB but the number of users vary from $5$ to $15$. All
channels are generated following the distribution
$\mathcal{CN}(0,1)$. Figure \ref{simu:2b} shows the optimal relay
power price, the total relay power sold, and the maximum relay revenue
with different numbers of users. We can see that as the number of
users increases, the optimal relay power price increases, the relay
power actually sold increases, and the maximum relay revenue
increases. Figure \ref{simu:2b} verifies the same law as Figure
\ref{simu:2a}, which says, if the supply is unchanged and demand
increases, it leads to higher equilibrium price and quantity.

\subsection{Static Network with Path-Loss Channel Only}
In this subsection, we study a static network whose channels are
deterministic instead of random. The network has three users, one
relay, and three destinations. The relative positions of the nodes
are shown in Figure \ref{simu:model}, where the coordinates of User
1-3, the relay, and Destination 1-3 are (-15, 3), (-10, 0), (-5,
-3), (0, 0), and (5, 3), (5, 0), (5, -3), respectively. We consider
the path-loss effect of wireless channels only by assuming that the
channel gains are inverse proportional to the distance squared. In
Figure 7, User 1 is the farthest from its destination and thus has
the worst channel; while User 3 is the closest to its destination
and has the best channel. The power of the users is set to be $10$
dB and the power of the relay is set to be $15$ dB.

\begin{figure}
  \centering
  \includegraphics[width=0.65\textwidth]{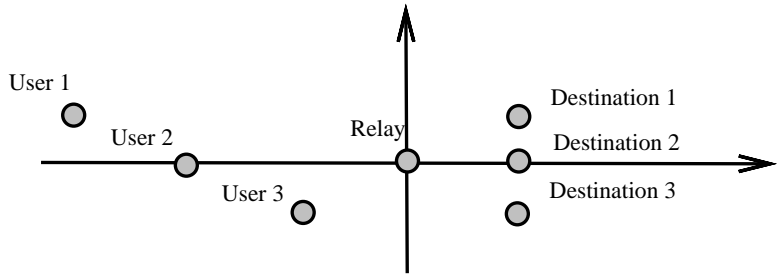}\\
  \label{fig-NWModel}
  \caption{A three-user static network.}\label{simu:model}
\end{figure}
\begin{figure}
  \centering
  \includegraphics[width=0.6\textwidth]{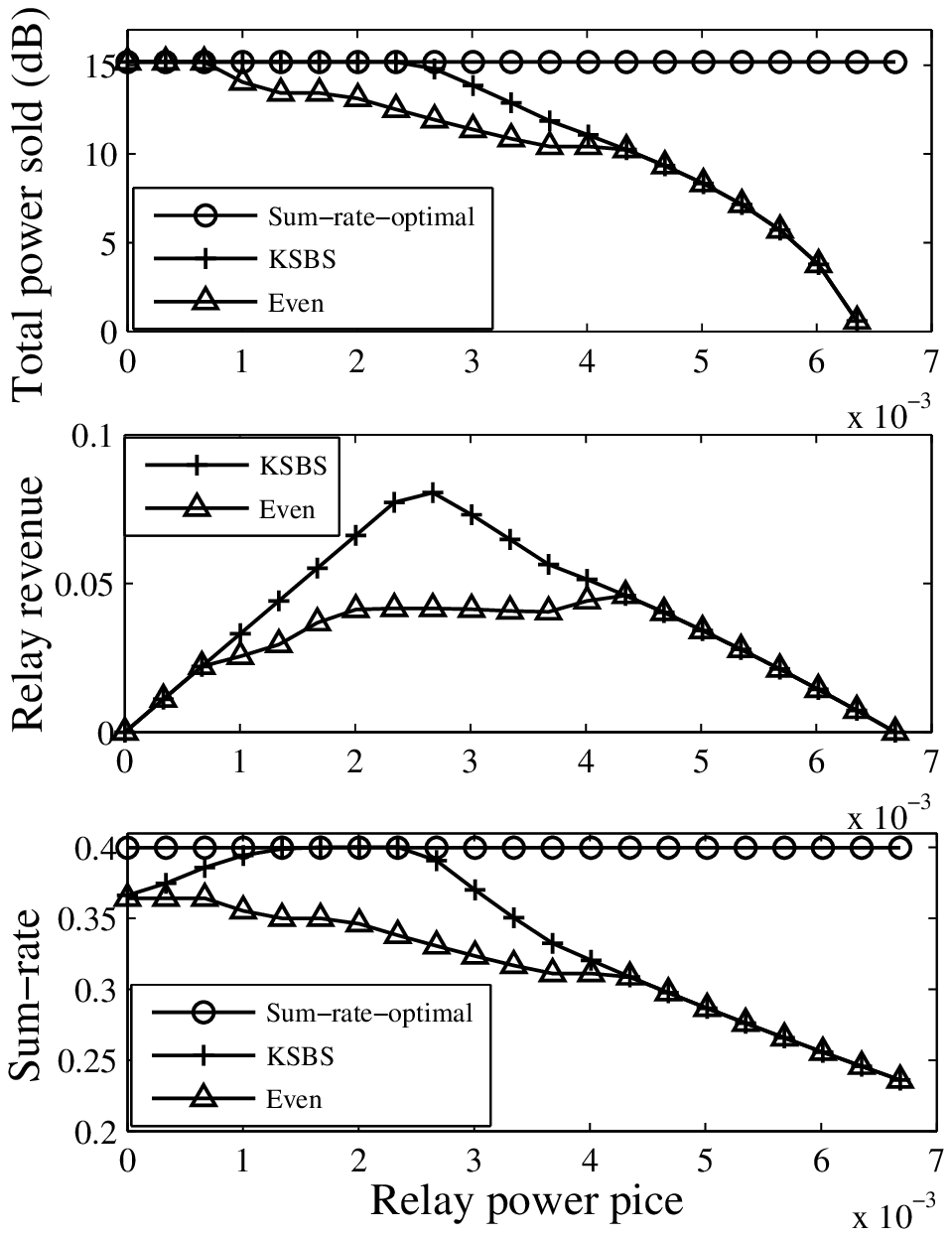}\\
  \caption{Power allocation, relay revenue, and sum-rate in three-user relay network with static channels.}\label{simu:3b}
\end{figure}

In Figure \ref{simu:3b}, the total power sold to the three users,
the relay revenue, and the network sum-rate are shown as the relay
power price varies. Three power allocation solutions are presented:
the proposed KSBS-based power allocation, the sum-rate-optimal power
allocation, and the even power allocation. Note that the
sum-rate-optimal allocation solution aims to maximize the network
sum-rate, is independent of the relay power price, and allocates all
the relay power $P$ to the three users. We can observe from Figure
\ref{simu:3b} that, when the price is higher, with the KSBS-based
and the even power allocation schemes, the users purchase less power
from the relay and the total power demand is smaller. For example,
using the KSBS-based allocation scheme, the total power demand is
less than $P$ when the price is higher than $0.0023$. Now let us
look at different price ranges separately. First, we can see that in
the price range $[0,0.0007]$, both KSBS-based and the even power
allocation schemes sell all relay power to the users. This is
because in this price range, $P_i^I(\lambda)\geq P/3$ for $i=1,2,3$,
thus with the even power allocation, each user will buy $P/3$, and
all relay power will be sold; for the KSBS-based power allocation,
$\sum^3_{i=1}P_i^I(\lambda)\geq P$, so all power of the relay will
be purchased by the users based on Lemma \ref{lem-sumpower}. Second,
when $\lambda \geq 0.0047$, the even power and the KSBS-based
schemes give the same power allocation results. This is because in
this price range, all three users' ideal power demands are no more
than $P/3$, that is, $P_i^I(\lambda) \leq P/3$ for $i=1,2,3$ and
$\sum^3_{i=1} P_i^I(\lambda)\leq P$. In this scenario, from Lemma
\ref{lem-sumpower}, both the even power allocation and the
KSBS-based schemes assign the ideal power demand $P_i^I(\lambda)$ to
User $i$, and the two schemes have the same performance. And when
$\lambda$ is in the range $[0.0007,0.0047]$, the KSBS-based power
allocation demands more relay power than the even power allocation,
and thus the relay receives a higher revenue in this range. This is
because with the even power allocation, a user cannot request more
than $1/3$ of the total relay power, while the KSBS-based scheme
does not have this constraint and thus enables users to request more
power. Furthermore, when $\lambda$ is $0.0027$ and the KSBS-based
scheme demands $91\%$ of the relay power to be sold to the users,
the relay revenue is maximized. At this relay power price, the
network sum-rate difference between the proposed KSBS-based solution
and the sum-rate-optimal power allocation is only about $2\%$. The
sum-rate difference between the even power and the sum-rate-optimal
schemes, however, is $23\%$ at the price $\lambda = 0.0047$, which
is the relay revenue maximizing price under the even power
allocation. For any relay price, the sum-rate difference between the
even power and the sum-rate-optimal schemes is no less than $9\%$.

\begin{table}
\caption{Achievable rates, normalized rate-difference, and the
system sum-rate in three user relay network with static
channels}\label{table:3b} \scriptsize{ \centering
\begin{tabular}{|c|c|c|c|c|c|c|c|c|c|c|c|c|}
 \hline
 \multicolumn{1}{|c|}{}&Sum-rate&\multicolumn{2}{|c|}{0}& \multicolumn{2}{|c|}{0.0013}& \multicolumn{2}{|c|}{0.0027}& \multicolumn{2}{|c|}{0.0047}& \multicolumn{2}{|c|}{0.0053}\\
\cline{3-12}
\multicolumn{1}{|c|}{}&-optimal&Even&KSBS&Even&KSBS&Even&KSBS&Even&KSBS&Even&KSBS\\ 
\hline
 $r_1$ & 0.0356&0.0498&0.0499  & 0.0356& 0.0356 &0.0356& 0.0356 &0.0356&0.0356  &0.0356& 0.0356
  \\ 
\hline $r_2$ &0.0838 &0.1017 &0.0994  &0.1017& 0.0991&0.0823& 0.0823
& 0.0627  &0.0627  &0.0627&0.0627\\ 
 \hline
  $r_3$&0.2802&0.2127 &0.2169    & 0.2127 & 0.2643 & 0.2127 & 0.2727
&0.1992&0.1992 &0.1777& 0.1777 \\
\hline \scriptsize{Rate-difference}&0.8729&0.7658  &0.7701 &0.8325
&0.8652&0.8325 &0.8694&0.8211&0.8211&0.7995 &0.7995
\\ 
\hline \scriptsize{Sum-rate}&0.3997&0.3641  &0.3662 &0.3500&0.3989
&0.3306 &0.3907&0.2975
&0.2975&0.2760 &0.2760\\
\hline
\end{tabular}
}
\end{table}

To further compare the performance of the three schemes, Table
\ref{table:3b} shows each user's achievable rate, the normalized
rate-difference, and the network sum-rate with the three power
allocation schemes at the relay power prices $0$, $0.0013$,
$0.0027$, $0.0047$ and $0.0053$. Recall that a smaller normalized
rate-difference represents a fairer power allocation. As can be seen
from Table \ref{table:3b}, the proposed KSBS-based scheme achieves a
smaller normalized rate difference than the sum-rate-optimal
solution for all relay prices, while the network sum-rate difference
between these two schemes is small. It shows that the proposed
solution is fairer than the sum-rate-optimal solution with
comparable network sum-rate. In sum, Figure \ref{simu:3b} and Table
\ref{table:3b} show that for the simulated network, the proposed
KSBS-based power allocation and relay pricing solutions achieve
close-to-optimal sum-rate, in the mean while, they also maximize the
relay revenue and achieve fairness among users.

To compare the sum-rates of the proposed solutions and the
sum-rate-optimal solution, we show in Figure \ref{simu:3e} the
network sum-rate of the proposed relay pricing and power allocation
solutions as the relay power constraint $P$ varies. We can see that
when the relay power constraint $P$ is small, indicating high demand
and low supply, the sum-rate of the proposed solution is almost the
same as the optimal sum-rate of the network. As the relay power
constraint $P$ increases, indicating low demand and high supply, the
sum-rate difference between the proposed solutions and the
sum-rate-optimal solution increases. When the relay power is $25$
dB, the difference is about $6\%$. The optimal relay price, on the
other hand, decreases as $P$ increases. These verifies the same law
of supply and demand with Figure \ref{simu:1a}, which says, if
supply increases and demand remains unchanged, then it leads to
lower equilibrium price.

\begin{figure}
  \centering
  \includegraphics[width=0.6\textwidth]{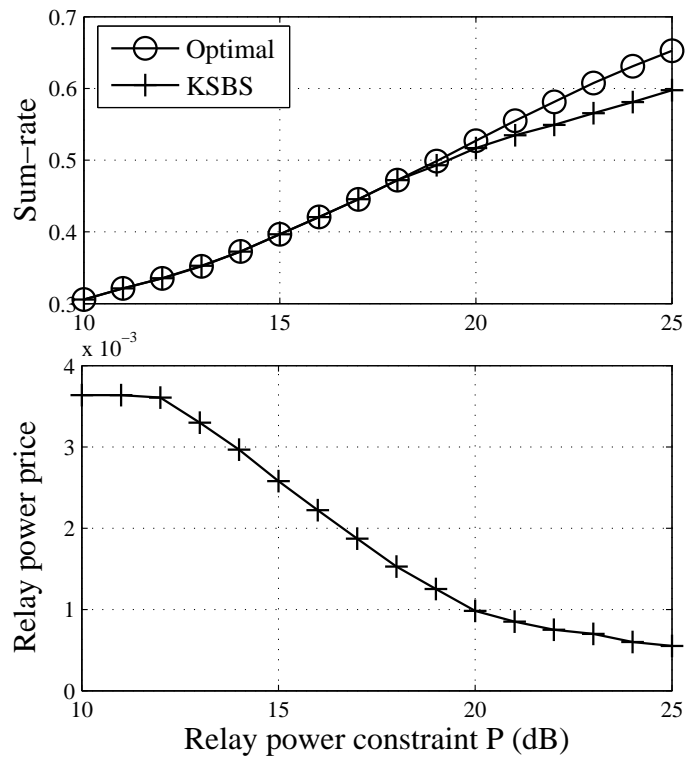}\\
  \caption{Sum-rate and relay power price in three user relay network with static channels.}\label{simu:3e}
\end{figure}
\section{Conclusion}\label{sec-conclusion}
In this paper, we study the relay power allocation problem in a
multi-user single-relay network. By introducing a relay power price,
we take into consideration the incentives for cooperation at the
relay in helping the users. The Stackelberg game is used to model
the interaction between the relay and the users, in which the relay
acts as the leader who sets the price of its power to gain the
maximum revenue and the users act as followers who pay for the relay
service. To model the competition among users, a bargaining game and
its KSBS are used for a fair power allocation. We analytically solve
the optimal relay price, while the problem of relay power allocation
among users is transformed into a convex optimization problem and
can be solved with efficient numerical methods. Simulation results
show that our solutions reflect the laws of supply and demand, give
better user utilities and relay revenue than even power allocation,
and approach the sum-rate-optimal power allocation in terms of
network sum-rate for a wide range of network scenarios.


%







\end{document}